\documentclass[11pt]{article}
\usepackage{ifpdf}
\usepackage[draft,uselot,hylinks,titlepage,full]{myarticle}

\newcommand{\dual}[1]{{{#1}^\perp}}
\newcommand{\dualk}[2]{{{#1}^\perp_{#2}}}

\newcommand{\ip}[1]{\angles{#1}}
\DeclareMathOperator{\supp}{{\mathrm supp}}

\newcommand{\dist}[1]{\Delta({#1})}

\renewcommand{\F}{{\mathbb{F}}}
\newcommand {\flipalg}{\mbox{{\sf Flip Algorithm}}\xspace}
\newcommand{\deff}{\rm{def}}
\newcommand{\codeword}{x}

\begin{document}

\title{LP decoding of expander codes: a simpler proof}
\author{Michael Viderman\thanks{Research supported by an ERC-2007-StG grant number 202405.} \\ Computer Science Department \\
Technion --- Israel Institute of Technology \\
Haifa 32000, Israel \\
{\tt viderman@cs.technion.ac.il}}

\maketitle

\begin{abstract}
A code $C \subseteq \F_2^n$ is a $(c,\epsilon,\delta)$-expander code if it has a Tanner graph, where every variable node has degree $c$, and every subset of variable nodes $L_0$ such that $|L_0|\leq \delta n$ has at least $\epsilon c |L_0|$ neighbors.

Feldman et al. (IEEE IT, 2007) proved that LP decoding corrects $\frac{3\epsilon-2}{2\epsilon-1} \cdot (\delta n-1)$ errors of $(c,\epsilon,\delta)$-expander code, where $\epsilon > \frac{2}{3}+\frac{1}{3c}$. 

In this paper, we provide a simpler proof of their result and show that this result holds for every expansion parameter $\epsilon > \frac{2}{3}$.
\end{abstract}

\section{Introduction}
The central algorithmic problem in coding theory is the explicit construction of error-correcting codes with best possible parameters together with fast encoding and decoding algorithms. Recently, this area has benefited
enormously from insights and viewpoints originating in complexity theory, and numerous interconnections between codes and complexity theory have been discovered, which are surveyed for example in \cite{Sud00,Tre04}. The
former survey \cite{Sud00} focuses on a notion called list-decoding and the second survey \cite{Tre04} focuses mainly on sub-linear algorithms for local testing and local decoding. Basically, there are two different kinds of
noise models: adversarial and probabilistic. In this paper we consider an \emph{adversarial noise model} where we only assume a bound on the number of errors and not how they are distributed. We refer a reader to the seminal paper of Richardson and Urbanke \cite{RichardsonU01} for details concerning \emph{probabilistic noise model}.

Surprisingly, all known constructions of asymptotically good error-correcting codes which can be decoded in linear time when a constant fraction of symbols is adversarially corrupted are based on expander codes.  
In particular, when the parity check graph of the low density parity check (LDPC) code has good expansion properties the associated code is called an expander code (see Definition \ref{def:expansion}). Expander codes and their decoding algorithms are often (implicitly) involved as basic building blocks in the constructions of asymptotically good codes which are linear-time decodable.

More formally, a linear code $C \subseteq \F_2^n$ is a $(c,\epsilon,\delta)$-expander code if it has a Tanner graph \cite{Tan81} $G=(L,R,E)$, where every variable node $l\in L$ has degree $c$ and every subset of variable nodes $L_0 \subseteq L$ such that $|L_0|\leq \delta n$ has at least $\epsilon c |L_0|$ neighbors.

The celebrated result of Sipser and Spielman \cite{SipSpielman} showed that regular expander codes with expansion parameter $\epsilon > 3/4$ are decodable in linear time from $(2\epsilon -1) \lfloor \delta n \rfloor$ errors. They achieved this by defining an extremely simple decoding algorithm, called \flipalg. Feldman et al. \cite{FMSSW07} showed that the linear programming (LP) decoding (suggested in \cite{FWK05}) corrects $\frac{3\epsilon-2}{2\epsilon-1} \cdot \lfloor \delta n \rfloor$ errors in {\em polynomial time} if the underlying code is a $(c,\epsilon,\delta)$-expander 
with expansion parameter $\epsilon > \frac{2}{3} + \frac{1}{3c}$. It is worth to notice that the running time of the LP decoder was strongly improved later (see e.g. \cite{Bur09,TSS09,Von08}), although linear running time was not achieved. Then it was shown that similar result can be obtained using a different decoding algorithm \cite{Vid12}, which in particular runs in linear time. Moreover, \cite[Proposition D.1]{Vid12} points out that some regular expander codes with expansion parameter $1/2$ have minimum Hamming distance 2 and thus cannot correct even a single error. This leaves open the intriguing question: whether all regular expander codes with expansion parameter more than $1/2$ are polynomial time decodable from the constant fraction of errors.

While the decoding algorithm provided in \cite{Vid12} reaches better (proven) result than LP decoding in \cite{FMSSW07} with regards to regular expander codes, the LP decoder still looks much more powerful. Hence we might hope that the future of studying of LP decoding capabilities (see \cite{ADS09,KoVo06a,KoVo06b,KoVo06c,HalE11}) will resolve the above question. 

The proof of the main result of Feldman et al. \cite{FMSSW07} was indirect and took the LP dual witness twice by using the min-cut max-flow theorem. 
We show a straightforward proof for the result of \cite{FMSSW07} which in particular holds for every expansion parameter $\epsilon > 2/3$. We think that our simplified proof sheds more light on the intuition lying behind the LP decoding of expander codes.

\section{Preliminaries}\label{sec:expanders}\label{sec:prelim}
Let $\F_2$ be the binary field and $[n]$ be the set $\set{1,\ldots,n}$. In this work, we consider only linear codes. Let $C\subseteq \F^n$ be a linear code over a field $\F$. For $w \in \F^n$, let $\supp(w)=\{i\in [n]~|~w_i \neq 0\}$ and $|w|=|\supp(w)|$. We define the \emph{distance} between two words $x,y \in \F^n$ to be $\dist{x,y} = |\{i\ |\ x_i \neq y_i\}|$. The minimum distance of a code is defined by $\displaystyle \dist{C} = \min_{x\neq y \in C} \dist{x,y}$. The dual code is defined by $\displaystyle \dual{C}=\left\{u \in \F^n\ |\ \forall c \in C:\ \ip{u,c}=0 \right\}$. For $T \subseteq \F^n$ we say that $w \perp T$ if for all $t\in T$ we have $\ip{w,t}=0$.

Now we define expanders and expander codes. We start from the definition of \emph{expanders} and then proceed to the definition of \emph{expander codes}.

\begin{definition}[Expander Codes]\label{def:expansion}
Let $C \subseteq \F_2^n$ be a linear code and let $G=(L,R,E)$ be its Tanner graph \cite{Tan81}, where $L=[n]$ represents the variable nodes and $R \subseteq \dual{C}$ represents the parity check nodes. Note that for 
every $\codeword \in \F_2^n$ we have $\codeword \perp R$ if and only if $\codeword \in C$. For $l\in L$ and $r\in R$ it holds that $\set{l,r} \in E$ if and only if $l \in \supp(r)$. 
For $T \subseteq L \cup R$ and $x\in L \cup R$, let
\begin{itemize}
\item $N(T) = \set{x_1\in L \cup R \ |\ \set{x_1,x_2}\in E \text{ \ for some \ } x_2 \in T}$ be the set of neighbors of $T$,

\item $N(x) = N(\set{x})$ be the set of neighbors of the node $x$.
\end{itemize}

Let $\epsilon,\delta > 0$ be constants. Then, $G$ is called a \emph{$(c,\epsilon,\delta)$-expander} if every vertex $l \in L$ has degree $c$ and for all subsets $S \subseteq L$ such that $|S|\leq \delta n$ we have $|N(S)| \geq \epsilon \cdot c|S|$.

We say that a code $C$ is a {\em $(c,\epsilon,\delta)$-expander code} if it has a parity check graph that is a $(c,\epsilon, \delta)$-expander.
\end{definition}

\section{Main Result}\label{sec:mainres}
Feldman et al. \cite{FMSSW07} proved the following theorem.

\begin{theorem}[\cite{FMSSW07}]\label{thm:feldman}
If $C \subseteq \F_2^n$ is a $(c,\epsilon,\delta)$-expander code, where $\epsilon > \frac{2}{3}+\frac{1}{3c}$ and $\epsilon c$ is an integer then $C$ is decodable from at most  $\frac{3\epsilon-2}{2\epsilon-1} \cdot (\lfloor \delta n \rfloor -1)$ errors (in polynomial time) by LP decoding.
\end{theorem}

Their proof took the LP dual witness twice and in particular, they applied min-cut max-flow theorem to find such a witness. Feldman et al. \cite{FMSSW07} mentioned that it would be interesting to see a more direct proof of this result.

We reprove their result in the more straightforward manner and obtain a small improvement to the required expansion parameter, i.e., we assume only that $\epsilon > 2/3$.

\begin{theorem}[Main Theorem]\label{thm:main}
If $C \subseteq \F_2^n$ is a $(c,\epsilon,\delta)$-expander code, where $\epsilon > \frac{2}{3}$ and $\epsilon c$ is an integer then $C$ is decodable from at most $\frac{3\epsilon-2}{2\epsilon-1} \cdot (\lfloor \delta n \rfloor-1)$ errors (in polynomial time) by LP decoding.
\end{theorem}

In Section \ref{sec:implicitres} we recall an implicit result (Theorem \ref{thm:implicit}) that was shown in \cite{FMSSW07}. We shall use Theorem \ref{thm:implicit} in the proof of Theorem \ref{thm:main}, which appears in Section \ref{sec:proofmainthm}.

\section{Implicit result of Feldman et al. \cite{FMSSW07}}\label{sec:implicitres}
Let $C \subseteq \F_2^n$ be a code and $G=(L=[n],R,E)$ be its parity check graph. Assume that $w\in \F_2^n$ is an input word (a corrupted codeword) and $U\subseteq [n]$ is a set of error coordinates, i.e., if all bits in $w|_U$ were flipped then $w$ would be the codeword of $C$. We say that $\gamma_i = -1$ if $i\in U$ and $\gamma_i = +1$ if $i\notin U$.

\begin{definition}[Feasible weights]\label{def:feasiblewei}
Let $G=(L,R,E)$ be the parity check graph of the code $C \subseteq \F_2^n$. A setting of edge weights $\displaystyle \set{\tau_{i,j}}_{(i,j)\in E}$ is called \emph{feasible} with respect to $U$ if
\begin{enumerate}
\item For all $j\in R$ and distinct $i,i'\in N(j)$ we have $\tau_{i,j} + \tau_{i',j} \geq 0$.

\item For all $i\in L$ we have $\sum_{j\in N(i)} \tau_{i,j} < \gamma_i$.
\end{enumerate}
\end{definition}

The following theorem was implicit in \cite{FMSSW07} and was argued in \cite[Proposition 2]{FMSSW07} and by the discussions in \cite[Sections 2 and 3]{FMSSW07}.

\begin{theorem}[Implicit in \cite{FMSSW07}]\label{thm:implicit}
If for all $U \subseteq [n]$, $|U| \leq \alpha n$ there exists a feasible settings of weights $\displaystyle \set{\tau_{i,j}}_{(i,j)\in E}$ then LP decoding for the code $C$ corrects $\alpha n$ errors. Moreover, if $C$ is an LDPC code then LP decoding runs in polynomial time.
\end{theorem}

The nice point here is that one can use Theorem \ref{thm:implicit} without any background on the LP decoding technique. We also note that this Theorem is related to the LDPC codes in general, not necessarily to the expander
codes.

\section{Proof of Theorem \ref{thm:main}}\label{sec:proofmainthm}
We start the proof by defining the concept of matching (slightly different and simpler than in \cite{FMSSW07}).

\begin{definition}[$q$-Matching]\label{def:matching}
Let $M \subseteq E$ be a subset and let $q$ be an integer. With some abuse of notations we say that $M$ is a matching if for every $j\in R$ we have at most one $i\in [n]$ such that $\set{i,j}\in M$. Given a subset $U\subseteq [n]$ we say that $M$ is a $q$-matching with respect to $U$ if
\begin{enumerate}
\item $M$ is a matching.

\item For every $u\in U$ we have at least $q$ nodes $j\in R$ such that $\set{u,j}\in M$.
\end{enumerate}
\end{definition}

For the rest of the proof, let $U \subseteq [n]$ be the set of error locations in the input word, and \[\hat{U} = \set{i\in [n]\setminus U \ |\ |N(i)\cap N(U)|\geq (2\epsilon-1) c}.\] We let $U' = U \cup \hat{U}$.
The following proposition is similar to the corresponding proposition in \cite{FMSSW07}, where an appropriate modification was made to fit our definition of matching (Definition \ref{def:matching}). 

\begin{proposition}\label{prop:feasweight}
Let $U' \subseteq [n]$ and assume $\epsilon > 2/3$. If there exists an $(\epsilon c)$-matching $M$ with respect to $U'$ then there is a feasible edge weight assignment (as in Definition \ref{def:feasiblewei}).
\end{proposition}
\begin{proof}
We assign weights as follows.
\begin{itemize}
\item For $j\in R$ such that there is an $i\in U$ and $(i,j)\in M$, set $\tau_{i,j}= -x$, and for all other $i'\in N(j)\setminus \set{i}$ set $\tau_{i',j}= +x$

\item For all other $j\in R$, for all $i\in N(j)$ set $\tau_{i,j}=0$.
\end{itemize}

Clearly, this weighting satisfies the first condition of a feasible weight assignment. For condition 2 we distinguish three cases. We recall that $\epsilon c$ is an integer. 

\begin{enumerate}
\item For $i\in U$ we have $\gamma_i = -1$. Also, at least $\epsilon c$ edges that are incident to $i$ are in $M$ (and each has weight $-x$). All other incident edges have weight $0$ or $+x$. So, the total weight of the incident edges to $i$ is at most $\epsilon c (-x) + (1-\epsilon)c(x) = cx(1-2\epsilon)$. This is less than $-1$ as long as $x > \frac{1}{(2\epsilon - 1)c}$.

\item If $i\in \hat{U}$ then $\gamma_i = +1$. At least $\epsilon c$ of $i$'s incident edges are in $M$, but (trivially) not incident to $U$. These edges have weight 0. The other incident edges have weight either 0 or $+x$.
So, the total weight is at most $(1-\epsilon) c x$ which is smaller than $+1$ as long as $x < \frac{1}{(1-\epsilon)c}$.

\item When $i\notin U \cup \hat{U}$ then $\gamma_i = +1$. We know that $i$ has less than $(2\epsilon-1) c$ neighbors in $N(U)$, i.e., at most $(2\epsilon-1) c-1$ neighbors in $N(U)$. All other edges have weight 0. I.e., the
total weight is at most $((2\epsilon-1) c -1) (x)$ which is smaller than $+1$ as long as $x < \frac{1}{(2\epsilon-1)c-1}$.
\end{enumerate}

Summarizing our conditions on $x$, we have $\frac{1}{(2\epsilon - 1)c} < x < \min\left\{\frac{1}{(1-\epsilon)c}, \frac{1}{(2\epsilon - 1)c-1}\right\}$. There is such an $x$ if $2\epsilon - 1 > 1-\epsilon$, i.e., $\epsilon > 2/3$.
\end{proof}

It is not hard to see that a sufficiently strong upper bound on $|U|$ implies an upper bound on $|U'|$.

\begin{proposition}\label{prop:boundset}
Assume that $|U| < \frac{3\epsilon-2}{2\epsilon-1} \cdot \lfloor \delta n \rfloor$. Then $|U'|=|U|+|\hat{U}| \leq \delta n$.
\end{proposition}
\begin{proof}
Recall that $\hat{U} = \set{i\in [n]\setminus U \ |\ |N(i)\cap N(U)| \geq (2\epsilon-1) c}$.

We prove that $|U \cup \hat{U}| < \lfloor \delta n \rfloor$. Assume the contrary and let $U'' \subseteq \hat{U}$ such that $|U \cup U''| = \lfloor \delta n \rfloor$ (note that $|U| < \lfloor \delta n \rfloor$).

We argue that $|U''| \leq \frac{1-\epsilon}{3\epsilon-2} |U|$ which implies that $|U \cup U''| = |U|+|U''| \leq \frac{2\epsilon-1}{3\epsilon-2} |U| < \lfloor \delta n \rfloor$ with contradiction. Due to the fact that $C$ is
a $(c,d,\epsilon,\delta)$-expander code we have
\[\epsilon c (|U|+|U''|) \leq |N(U \cup U'')| = |N(U)|+ |N(U'')\setminus N(U)| \leq |U|\cdot c + |U''|\cdot (c-(2\epsilon-1)c) .\]
Then we have $|U''| (3\epsilon -2)c \leq |U| (1-\epsilon)c$ and $|U''| \leq \frac{1-\epsilon}{3\epsilon-2} |U|$. Contradiction.
\end{proof}

We are ready to prove Theorem \ref{thm:main}.

\begin{proof}[Proof of Theorem \ref{thm:main}]
Let $U \subseteq [n]$ be a set of error locations. By assumption, $|U| \leq \frac{3\epsilon-2}{2\epsilon-1} \cdot \lfloor \delta n \rfloor$ and so, by Proposition \ref{prop:boundset} we have $|U'|=|U|+|\hat{U}| \leq \delta n$. By Corollary \ref{cor:hallthm} (stated below) there exists an $(\epsilon c)$-matching $M$ with respect to $U'$. By Proposition \ref{prop:feasweight} there exists a feasible edge weight assignment. Theorem \ref{thm:implicit} implies that the LP decoder succeeds and runs in polynomial time.
\end{proof}

\subsection{Hall's Theorem}
Let us recall that $G=(L,R,E)$ is the bipartite graph with vertex sets $L$ and $R$, and the edge set $E$. We also recall Definition \ref{def:matching}.
We let $X \subseteq L$ to be a subset of $L$.

\begin{theorem}[Hall's Marriage Theorem]\label{thm:hallthm}
Assume that for every $S \subseteq X$ we have $\vert N(S)\vert \geq \vert S\vert$. Then there exists a $1$-matching with respect to $X$.
\end{theorem}

\begin{corollary}[Polygamous form of Hall's Theorem]\label{cor:hallthm}
Let $d \geq 1$ be an integer. Assume that for all $S \subseteq X$ we have $\vert N(S)\vert \geq d\vert S\vert$.
Then there exists $d$-matching with respect to $X$.
\end{corollary}
\begin{proof}
Replace each $x \in X$ with $d$ nodes connected to all nodes in $N(x)$. The corollary follows from the Hall's theorem (Theorem \ref{thm:hallthm}).
\end{proof}

\subsection*{Acknowledgements}
We thank the anonymous referees for valuable comments on an earlier version of this paper.  

\bibliography{LP_decoding}
\bibliographystyle{plain}

\end{document}